\let\accentvec\vec
\documentclass[11pt,a4paper]{llncs}

\let\vec\accentvec
\usepackage{amsmath}
\usepackage{amssymb}

\usepackage{enumerate}
\usepackage{graphicx}
\usepackage{epstopdf}
\usepackage{latexsym}
\usepackage{MnSymbol}
\usepackage{mathrsfs}
\usepackage{courier}
\usepackage[numbers,square,sort&compress]{natbib}

\usepackage{llncsOmat}

\usepackage[a4paper,margin=1in,includefoot]{geometry}

\begin{document}

\thispagestyle{empty}

\title{Synthesizing Minimal Tile Sets for Patterned DNA Self-Assembly}
\author{Mika G\"o\"os \and Pekka Orponen}
\institute{Department of Information and Computer Science \\
Aalto University School of Science and Technology (TKK) \\
P.O. Box 15400, FI-00076 Aalto, Finland\\
\email{\{mika.goos, pekka.orponen\}@tkk.fi}}

\maketitle

\begin{abstract}
The Pattern self-Assembly Tile set Synthesis (PATS) problem is to
determine a set of coloured tiles that self-assemble to implement a
given rectangular colour pattern. We give an exhaustive
branch-and-bound algorithm to find tile sets of minimum cardinality
for the PATS problem. Our algorithm makes use of a search tree in the
lattice of partitions of the ambient rectangular grid, and an efficient
bounding function to prune this search tree. Empirical data on the
performance of the algorithm shows that it compares favourably to
previously presented heuristic solutions to the problem.
\end{abstract}

\section{Introduction}
\label{sec:introduction}

\subsection{Background}

An appealing methodology for bottom-up manufacturing of nanoscale
structures and devices is to use a self-assembling system of DNA
tiles~\citep{Rothemund2006} to build a scaffold structure on which
functional units are deposited~\citep{Lin2006,Park2004,Yan2003}.
A systematic approach to the design of self-assembling DNA scaffold
structures was proposed and experimentally validated by Park et al.\
in~\citep{Park2006}. However, as pointed out by Ma \& Lombardi
in~\citep{Ma2008}, that design is wasteful of tile types, i.e.\
generally speaking the same scaffold structures can be assembled also
from fewer types of specially-manufactured DNA complexes, thus
reducing the requisite laboratory work.

Ma \& Lombardi~\citep{Ma2008} formulated the task of minimizing the
number of DNA tile types required to implement a given 2-D pattern
abstractly as a combinatorial optimization problem,
the \emph{Patterned self-Assembly Tile set Synthesis} (PATS) problem,
and proposed two greedy heuristics for solving it. In this paper, we
present a systematic branch-and-bound approach to exploring the space
of feasible PATS tilings, and assess its computational performance
also experimentally. The method compares favourably to the heuristics
proposed by Ma \& Lombardi, finding noticeably smaller or even
provably minimal tile sets in a reasonable amount of computation time.
However, as the experimental results in Section~\ref{sec:results} show,
the computational problem still remains quite challenging for large patterns.

\subsection{Overview}

Our considerations take place in the \emph{abstract Tile Assembly
  Model} (aTAM) of Winfree and
Rothemund~\citep{Rothemund2001,Rothemund2000,Winfree1998}. The key
features of this model are as follows. The basic building blocks to
construct different two-dimensional shapes are called \emph{tiles}. A
tile is a non-rotating unit square that has different kinds of
\emph{glues} of varying strengths associated with each of its four
edges. Only finitely many different combinations of glues are allowed,
that is, we consider only finitely many \emph{glue types}; we have an
unlimited supply of tiles, however. In the aTAM, tile assemblies are
constructed through the process of self-assembly. Self-assembly begins
with a \emph{seed assembly}, an initial tile assembly, already in
place in the discrete grid $\Z\times\Z$. A new tile can extend an
existing assembly by binding itself into a suitable position in the
grid: we require a new tile to interact with existing tiles with
sufficient strength to overcome a certain universal \emph{temperature}
threshold. A more detailed account of the aTAM model is given in
section \ref{sec:atam}.

In the PATS problem~\citep{Ma2008}, one associates a \emph{colour}
with each tile type and targets a specific coloured \emph{pattern}
within a rectangular assembly. The question is: given the desired
colour pattern, what is the smallest set of (coloured) tile types that
will self-assemble to implement it? The specifics of the PATS problem are
given in section \ref{sec:pats}.

Our definition of the PATS problem restricts the self-assembly process
to proceed in a uniform way. This simplification allows us to design
efficient strategies for an exhaustive search. In
section \ref{sec:algorithm} we give full particulars of a novel
branch-and-bound (B\&B) algorithm for the PATS problem. For a pattern
of size $m\times n$, we reduce the problem of finding a minimal tile
set to the problem of finding a minimum-size \emph{constructible}
partition of $[m]\times[n]$. Here, constructibility of a partition can
be verified in time polynomial in $m$ and $n$. This leads us to
construct a search tree in the lattice of partitions of the set
$[m]\times[n]$ and to find pruning strategies for this search tree.
In the concluding sections \ref{sec:results} and \ref{sec:conclusion}
we give some performance data on the B\&B algorithm and summarize
our contributions.

\section{Preliminaries}
\label{sec:preliminaries}

\subsection{The abstract tile assembly model}
\label{sec:atam}

Our notation is derived from those of \citep{Adleman2002,Lathrop2009,Rothemund2000}. First, to simplify our notations, let $\mathcal{D} = \{N,E,S,W\}$ be the set of four functions $\Z^2\to\Z^2$ corresponding to the cardinal directions (north, east, south, west) so that $N(x,y)=(x,y+1)$, $E(x,y)=(x+1,y)$, $S=N^{-1}$ and $W=E^{-1}$.

Let $\Sigma$ be a set of \emph{glue types} and $s:\Sigma\times\Sigma \to \mathbb{N}$ a \emph{glue strength} function such that $s(\sigma_1,\sigma_2)=s(\sigma_2,\sigma_1)$ for all $\sigma_1,\sigma_2\in\Sigma$. In this paper, we only consider glue strength functions for which $s(\sigma_1,\sigma_2)=0$ if $\sigma_1\neq\sigma_2$. A \emph{tile type} $t\in\Sigma^4$ is a quadruple $(\sigma_N(t), \sigma_E(t), \sigma_S(t), \sigma_W(t))$ of glue types for each side of a unit square. Given a set $\Sigma$ of glues, an \emph{assembly} $\mathcal{A}$ is a partial mapping from $\Z^2$ to $\Sigma^4$. A \emph{tile assembly system} (TAS) $\mathscr{T}=(T,\mathcal{S},s,\tau)$ consists of a finite set $T$ of tile types, an assembly $\mathcal{S}$ called the \emph{seed assembly}, a glue strength function $s$ and a \emph{temperature} $\tau \in \Z^+$ (we use $\tau = 2$).

To formalize the self-assembly process, we first fix a TAS $\mathscr{T}=(T,\mathcal{S},s,\tau)$. For two assemblies $\mathcal{A}$ and $\mathcal{A}'$ we write $\mathcal{A} \rightarrow_{\mathscr{T}} \mathcal{A}'$ if there exists a pair $(x,y)\in\Z^2$ and a tile $t\in T$ such that $\mathcal{A}' = \mathcal{A} \cup \{((x,y),t)\}$, where the union is disjoint, and
\begin{equation}
\sum_D s(\sigma_D(t),\sigma_{D^{-1}}(\mathcal{A}(D(x,y)))\enspace\geq\enspace\tau\enspace,
\end{equation}
where $D$ ranges over those directions in $\mathcal{D}$ for which $\mathcal{A}(D(x,y))$ is defined. This is to say that a new tile can be adjoined to an assembly $\mathcal{A}$ if the new tile shares a common boundary with tiles that bind it into place with total strength at least $\tau$.

Let $\rightarrow^*_{\mathscr{T}}$ be the reflexive transitive closure of $\rightarrow_{\mathscr{T}}$.  A TAS $\mathscr{T}$ \emph{produces} an assembly $\mathcal{A}$ if $\mathcal{A}$ is an \emph{extension} of the seed assembly $\mathcal{S}$, that is if $\mathcal{S}\rightarrow^*_{\mathscr{T}}\mathcal{A}$. Let us denote by $\Prod \mathscr{T}$ the set of all assemblies produced by $\mathscr{T}$. This way, the pair $(\Prod \mathscr{T}, \rightarrow^*_{\mathscr{T}})$ forms a partially ordered set. We say that a TAS $\mathscr{T}$ is \emph{deterministic} if for any assembly $\mathcal{A} \in \Prod \mathscr{T}$ and for every $(x,y)\in\Z^2$ there exists at most one $t\in T$ such that $\mathcal{A}$ can be extended with $t$ at position $(x,y)$. A TAS $\mathscr{T}$ is deterministic precisely when $\Prod \mathscr{T}$ is a lattice. Also, the maximal elements in $\Prod \mathscr{T}$ are such assemblies $\mathcal{A}$ that can not be further extended, that is, there do not exist assemblies $\mathcal{A}'$ such that $\mathcal{A}\rightarrow_{\mathscr{T}}\mathcal{A}'$. These maximal elements are called \emph{terminal assemblies}. We denote by $\Term \mathscr{T}$ the set of terminal assemblies of $\mathscr{T}$. If all \emph{assembly sequences}
\begin{equation}
\mathcal{S} \rightarrow_{\mathscr{T}} \mathcal{A}_1 \rightarrow_{\mathscr{T}} \mathcal{A}_2 \rightarrow_{\mathscr{T}} \cdots
\end{equation}
terminate and $\Term \mathscr{T} = \{\mathcal{P}\}$ for some assembly $\mathcal{P}$, we say that $\mathscr{T}$ \emph{uniquely produces} $\mathcal{P}$.

\subsection{The PATS problem}
\label{sec:pats}

In this paper we restrict our attention to designing minimal tile assembly systems that construct a given pattern in a finite rectangular $m$ by $n$ grid $[m]\times[n]\subseteq\Z^2$. This problem was first discussed by Ma \& Lombardi \citep{Ma2008}.

A mapping from $[m]\times[n]$ onto $[k]$ defines a \emph{$k$-colouring} or a \emph{$k$-coloured pattern}. To build a given pattern, we start with boundary tiles in place for the west and south borders of the $m$ by $n$ rectangle and keep extending this assembly by tiles with strength-1 glues.

\begin{definition}[Pattern self-Assembly Tile set Synthesis (PATS) \normalfont\citep{Ma2008}\bfseries]
\vspace{.3cm}

\begin{tabular}[t]{r@{\hspace{.3cm}}p{0.8\textwidth}}
\textbf{Given:} & A $k$-colouring $c:[m]\times[n]\to[k]$. \\
\textbf{Find:}  & A tile assembly system $\mathscr{T} = (T,\mathcal{S},s,2)$ such that
\begin{enumerate}[P1. ]
		\item The tiles in $T$ have bonding strength 1.
		\item The domain of $\mathcal{S}$ is $[0,m]\times\{0\}\cup\{0\}\times[0,n]$ and all the terminal assemblies have the domain $[0,m]\times[0,n]$.
		\item There exists a colouring $d:T\to[k]$ such that for each terminal assembly $\mathcal{A}\in\Term\mathscr{T}$ we have $d(\mathcal{A}(x,y))=c(x,y)$ for all $(x,y)\in[m]\times[n]$.
	\end{enumerate} \\
\end{tabular}
\end{definition}

In particular, we are interested in the \emph{minimal solutions} (in terms of $|T|$) to the PATS problem. By the same token, we can make the following assumption:
\begin{assumption}
In our TASs, every tile participates in assembling some terminal assembly.
\end{assumption}

Ma \& Lombardi show a certain derivative of the above optimization problem \textsf{NP}-hard in \citep{Ma2009}. However, to our knowledge, a proof of the \textsf{NP}-hardness of the PATS problem as stated above is lacking.

As an illustration, we construct a part of the Sierpinski triangle with a 4-tile TAS in Figure \ref{fig:sierpinski}. We use natural numbers as glue labels in our figures.

\begin{figure}[t]
  \centering
  \includegraphics[width=\textwidth]{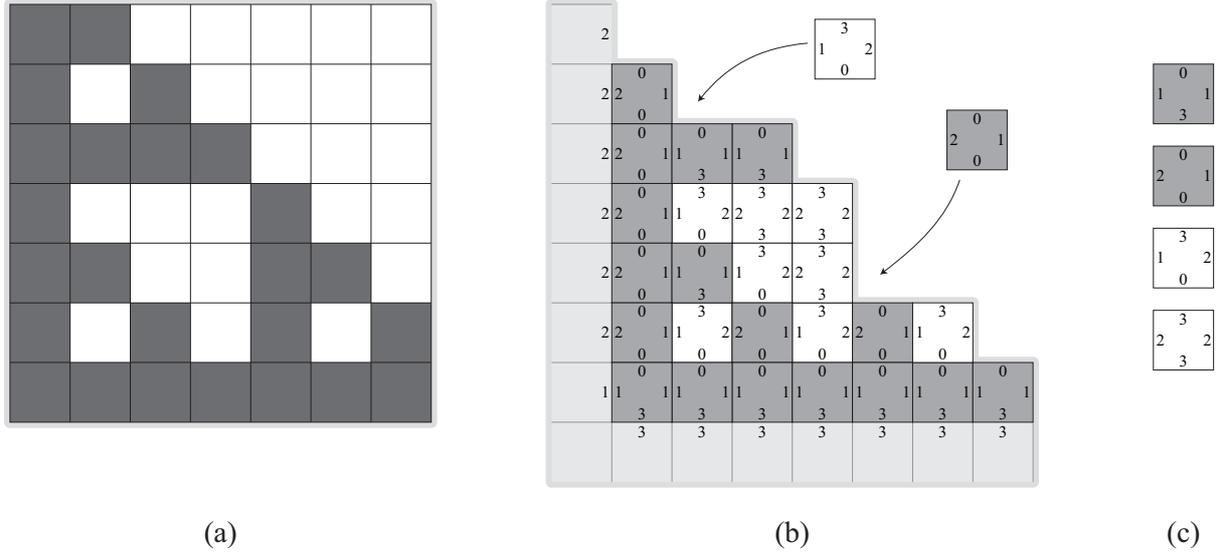}
  \caption{(a) A finite subset of the discrete Sierpinski triangle. This 2-colouring of the set $[7]\times[7]$ defines an instance of the PATS problem. (b) Assembling the Sierpinski pattern with a TAS that has an appropriate seed assembly and a (coloured) tile set shown in (c).}
  \label{fig:sierpinski}
\end{figure}

In the literature, the seed assembly of a TAS is often taken to be a single seed tile \citep{Adleman2002,Rothemund2000} whereas we consider an L-shaped seed assembly. These boundaries can always be self-assembled using $m+n+1$ different tiles with strength-2 glues, but for practical purposes we allow for the possibility of using, for example, DNA origami techniques \citep{Fujibayashi2008} to construct these boundary conditions.

Due to constraint \textit{P1} the self-assembly process proceeds in a uniform manner directed from south-west to north-east. This paves the way for a simple characterization of deterministic TASs in the context of the PATS problem.
\begin{proposition} \label{prop:determinism}
Solutions $\mathscr{T} = (T,\mathcal{S},s,2)$ of the PATS problem are deterministic precisely when for each pair of glue types $(\sigma_1,\sigma_2)\in\Sigma^2$ there is at most one tile type $t\in T$ so that $\sigma_S(t)=\sigma_1$ and $\sigma_W(t)=\sigma_2$.
\end{proposition}

A simple observation reduces the work needed in finding minimal solutions of the PATS problem.
\begin{lemma}
The minimal solutions of the PATS problem are deterministic TASs.
\end{lemma}
\begin{proof}
For the sake of contradiction, suppose that $\mathscr{N} = (T,\mathcal{S},s,2)$ is a minimal solution to a PATS problem instance and that $\mathscr{N}$ is not deterministic. By the above proposition let tiles $t_1,t_2\in T$ be such that $\sigma_S(t_1)=\sigma_S(t_2)$ and $\sigma_W(t_1)=\sigma_W(t_2)$. Consider the simplified TAS $\mathscr{N}' = (T\smallsetminus\{t_2\},\mathcal{S},s,2)$. We show that this, too, is a solution to the PATS problem, which violates the minimality of $|T|$.

Suppose $\mathcal{A}\in\Term\mathscr{N}'$. If $\mathcal{A}\notin\Term\mathscr{N}$, then some $t\in T$ can be used to extend $\mathcal{A}$ in $\mathscr{N}$. If $t\in T\smallsetminus\{t_2\}$, then $t$ could be used to extend $\mathcal{A}$ in $\mathscr{N}'$, so we must have $t=t_2$. But since new tiles are always attached by binding to south and west sides of the tile, $\mathcal{A}$ could then be extended by $t_1$ in $\mathscr{N}'$. Thus, we conclude that $\mathcal{A}\in\Term\mathscr{N}$ and furthermore $\Term\mathscr{N}'\subseteq\Term\mathscr{N}$. This demonstrates that $\mathscr{N}'$ has property \textit{P2}. The properties \textit{P1} and \textit{P3} can be readily seen to hold for $\mathscr{N}'$ as well. In terms of $|T|$ we have found a more optimal solution---and a contradiction.
\qed
\end{proof}
\begin{assumption}
We consider only deterministic TASs in the sequel.
\end{assumption}

\section{A branch-and-bound algorithm}
\label{sec:algorithm}

We describe an exact algorithm to find minimal solutions to the PATS
problem. We extend the methods of \citep{Ma2008} to obtain an
exhaustive branch-and-bound (B\&B) algorithm.
The idea of Ma \& Lombardi \citep{Ma2008} (following experimental work
of \citep{Park2006}) is to start with an \emph{initial tile set} that
consists of $m\cdot n$ different tiles, one for each of the grid
positions in $[m]\times[n]$. Their algorithm then proceeds to merge
tile types in order to minimize $|T|$. We formalize this search
process as an exhaustive search in the set of all partitions of the
set $[m]\times[n]$.
In the following, we let a PATS instance be given by a fixed
$k$-coloured pattern $c:[m]\times[n]\to [k]$.

\subsection{The search space}

Let $X$ be the set of partitions of the set $[m]\times[n]$. For partitions $P,P'\in X$ we define a relation $\sqsubseteq$ so that
\begin{equation}
P \sqsubseteq P'\quad\Longleftrightarrow\quad \forall p'\in P' :\enspace\exists p\in P:\enspace p'\subseteq p\enspace.
\end{equation}
Now, $(X,\sqsubseteq)$ is a partially ordered set, and in fact, a lattice. If $P \sqsubseteq P'$ we say that $P'$ is a \emph{refinement} of $P$, or that $P$ is \emph{coarser} than $P'$.
Note that $P\sqsubseteq P'$ implies $|P|\leq |P'|$.

The colouring $c$ induces a partition $P(c)=\{c^{-1}(\{i\})\enspace|\enspace i\in [k]\}$ of the set $[m]\times[n]$.
In addition, since every (deterministic) solution $\mathscr{T} = (T,\mathcal{S},s,2)$ of the PATS problem uniquely produces some assembly $\mathcal{A}$, we associate with $\mathscr{T}$ a partition $P(\mathscr{T}) = \{\mathcal{A}^{-1}(\{t\})\enspace|\enspace t\in\mathcal{A}([m]\times[n])\}$. Here, $|P(\mathscr{T})|=|T|$ due to our Assumptions 1 and 2. With this terminology, the condition \textit{P3} in the definition of the PATS problem is equivalent to requiring that a TAS $\mathscr{T}$ satisfies
\begin{equation}
P(c) \sqsubseteq P(\mathscr{T})\enspace.
\end{equation}

We say that a partition $P\in X$ is \emph{constructible} if $P=P(\mathscr{T})$ for some deterministic TAS $\mathscr{T}$ with properties \textit{P1} and \textit{P2}. With this, we can rephrase our goal from the point of view of using partitions as the fundamental search space.
\begin{proposition}
A minimal solution to the PATS problem corresponds to a partition $P\in X$ such that $P$ is constructible, $P(c) \sqsubseteq P$ and $|P|$ is minimal.
\end{proposition}

For example, the 2-coloured pattern in Figure \ref{fig:pattern}a defines a 2-part partition, $A$, say. The 7-part partition $M$ in Figure \ref{fig:pattern}b is a refinement of $A$ ($A\sqsubseteq M$) and in fact, $M$ is constructible (see Figure \ref{fig:construction}b) and corresponds to a minimal solution of the PATS problem defined by the pattern $A$.

\begin{figure}[t]
  \centering
  \includegraphics[width=.6\textwidth]{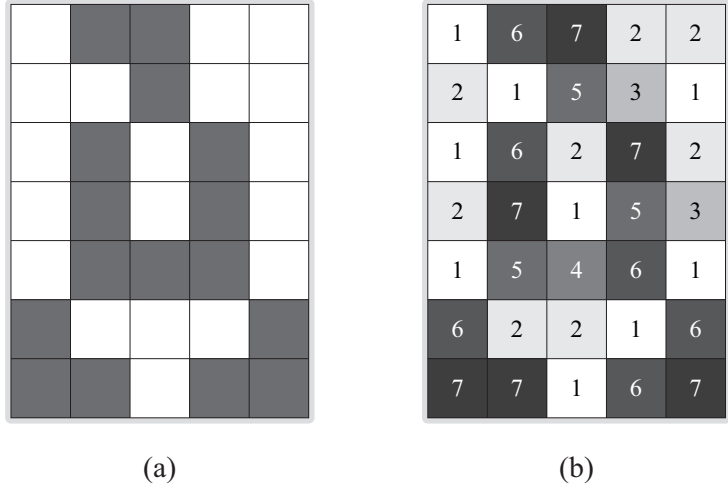}
  \caption{(a) Partition $A$. (b) A partition $M$ that is a refinement of $A$ with $|M|=7$ parts.}
  \label{fig:pattern}
\end{figure}

\subsection{Determining constructibility}
In this section we give an algorithm for deciding the constructibility of a given partition in polynomial time. To do this, we use the concept of most general (or least constraining) tile assignments. For simplicity, we assume the set of glue labels $\Sigma$ to be infinite.

\begin{definition}
Given a partition $P$ of the set $[m]\times[n]$, a \emph{most general tile assignment} (MGTA) is a function $f:P\to\Sigma^4$ such that
\begin{enumerate}[{A}1. ]
	\item When every position in $[m]\times[n]$ is assigned a tile type according to $f$, any two adjacent positions agree on the glue type of the side between them.
	\item For all assignments $g:P\to\Sigma^4$ satisfying \textit{A1} we have\footnote{To shorten the notation we write $f(p)_D$ instead of $\sigma_D(f(p))$.}
	\begin{equation}
	f(p_1)_{D_1}=f(p_2)_{D_2} \quad \implies \quad g(p_1)_{D_1}=g(p_2)_{D_2}
	\end{equation}
	for all $(p_1,D_1),(p_2,D_2)\in P\times\mathcal{D}$.
\end{enumerate}
\end{definition}

To demonstrate this concept we present a most general tile assignment $f:I\to\Sigma^4$ for the \emph{initial partition} $I=\{\{a\}\enspace|\enspace a\in[m]\times[n]\}$ in Figure \ref{fig:construction}a and a MGTA for the partition of Figure \ref{fig:pattern}b in Figure \ref{fig:construction}b.

\begin{figure}[t]
  \centering
  \includegraphics[width=\textwidth]{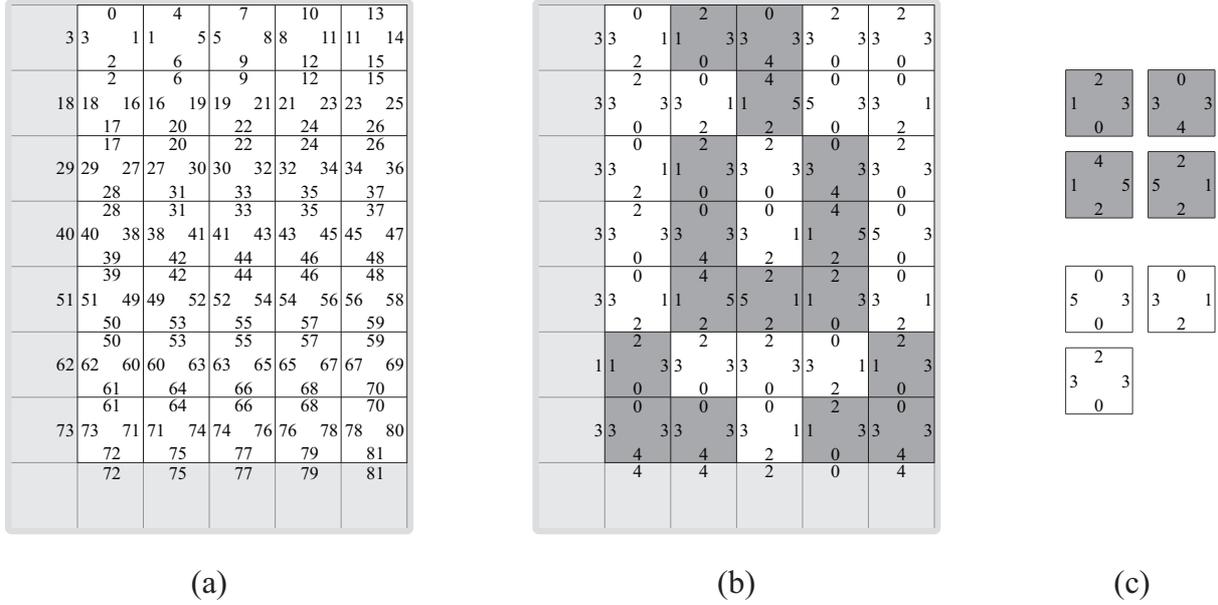}
  \caption{(a) A MGTA for the constructible initial partition $I$ (with a seed assembly in place). (b) Finished assembly for the pattern from Figure \ref{fig:pattern}a. The tile set to construct this assembly is given in (c).}
  \label{fig:construction}
\end{figure}

Given a partition $P\in X$ and a function $f:P\to\Sigma^4$, we say that $g:P\to\Sigma^4$ is obtained from $f$ by \emph{merging glues $a$ and $b$} if for all $(p,D)\in P\times\mathcal{D}$ we have
\begin{equation}
g(p)_D=\left\{
\begin{aligned}
a,\qquad&\text{if }f(p)_D=b\\
f(p)_D,\qquad&\text{otherwise}
\end{aligned}\right.\quad.
\end{equation}

A most general tile assignment for a partition $P\in X$ can be found as follows. We start with a function $f_0:P\to\Sigma^4$ that assigns to each tile edge a unique glue type, or in other words, a function $f_0$ so that the mapping $(p,D)\mapsto f_0(p)_D$ is injective. Next, we go through all pairs of adjacent positions in $[m]\times[n]$ (in some order) and require their matching sides to have the same glue type by merging the corresponding glues. This process generates a sequence of functions $f_0,f_1,f_2,\ldots,f_N=f$ and terminates after $N\leq 2mn$ steps.

\begin{lemma} \label{le:mgta}
The above algorithm generates a most general tile assignment.
\end{lemma}
\begin{proof}
By the end, we are left with a function $f$ that satisfies property \textit{A1} by construction. To see why property \textit{A2} is satisfied, we again use the language of partitions.

Any assignment gives rise to a set of equivalence classes (or a partition) on $P\times\mathcal{D}$: Elements that are assigned the same glue type reside in the same equivalence class. The initial assignment $f_0$ gives each part-direction pair a unique glue type, and thus, corresponds to the initial partition $J=\{\{a\}\enspace|\enspace a\in P\times\mathcal{D}\}$. In the algorithm, any glue merging operation corresponds to the combination of two equivalence classes.

The algorithm goes through a list of pairs $\{\{a_i,b_i\}\}_{i=0}^{N-1}$ of elements from $P\times\mathcal{D}$ that are required to have the same glue type. In this way, the list records necessary conditions for property \textit{A1} to hold. This is to say that every assignment satisfying \textit{A1} has to correspond to a partition that is coarser than each of the partitions in $\mathcal{L}=\{J[a_i,b_i]\}_{i=0}^{N-1}$, where $J[a,b]$ is the partition obtained from the initial partition by combining parts $a$ and $b$. Since the set $(P\times\mathcal{D}, \sqsubseteq)$ is a lattice, there exists a unique greatest lower bound $\inf \mathcal{L}$ of the partitions in $\mathcal{L}$. This is exactly the partition that the algorithm calculates in the form of the assignment $f$. As a greatest lower bound, $\inf \mathcal{L}$ is finer than any partition corresponding to an assignment satisfying \textit{A1}, but this is precisely the requirement for condition \textit{A2}.\qed
\end{proof}

The above analysis also gives the following.
\begin{corollary}
For a given partition, MGTAs are unique up to relabeling of the glue types.
\end{corollary}
Thus, for each partition $P$, we take \emph{the MGTA for $P$} to be some canonical representative from the class of MGTAs for $P$.

For efficiency purposes, it is worth mentioning that MGTAs can be generated iteratively: A partition $P\in X$ can be obtained by repeatedly combining parts starting from the initial partition $I$:
\begin{equation}
I = P_1 \sqsupseteq P_2 \sqsupseteq \cdots \sqsupseteq P_N = P\enspace.
\end{equation}
As a base case, a MGTA for $I$ can be computed by the above algorithm. A MGTA for each $P_{i+1}$ can be computed from a MGTA for the previous partition $P_i$ by just a small modification: Let a MGTA $f_i:P_i\to\Sigma^4$ be given for $P_i$ and suppose $P_{i+1}$ is obtained from $P_i$ by combining parts $p_1,p_2\in P_i$. Now, a MGTA $f_{i+1}$ for $P_{i+1}$ can be obtained from $f_i$ by \emph{merging tiles} $f_i(p_1)$ and $f_i(p_2)$, that is, merging the glue types on the four corresponding sides.

We now give the conditions for a partition to be constructible in terms of MGTAs.
\begin{lemma}
A partition $P\in X$ is constructible iff the MGTA $f:P\to\Sigma^4$ for $P$ is injective and the tile set $f(P)$ is deterministic in the sense of Proposition \ref{prop:determinism}.
\end{lemma}
\begin{proof}
``$\Rightarrow$'': Let $P\in X$ be constructible and let the MGTA $f:P\to\Sigma^4$ for $P$ be given. Let $\mathscr{T}$ be a deterministic TAS such that $P(\mathscr{T})=P$. The uniquely produced assembly of $\mathscr{T}$ induces a tile assignment $g:P\to\Sigma^4$ that satisfies property \textit{A1}. Now using property \textit{A2} for the MGTA $f$ we see that any violation of the injectivity of $f$ or any violation of the determinism of the tile set $f(P)$ would imply such violations for $g$. But since $g$ corresponds to a constructible partition, no violations can occur for $g$ and thus none for $f$.

``$\Leftarrow$'': Let $f:P\to\Sigma^4$ be an injective MGTA with deterministic tile set $f(P)$. Because $f(P)$ is deterministic, we can choose glue types for a seed assembly $\mathcal{S}$ so that the westernmost and southernmost tiles fall into place according to $f$ in the self-assembly process. The TAS $\mathscr{T}=(f(P), \mathcal{S},s,2)$, with appropriate glue strengths $s$, then uniquely produces a terminal assembly that agrees with $f$ on $[m]\times[n]$. This gives $P(\mathscr{T}) \sqsubseteq P$, but since $f$ is injective, $|P|=|f(P)|=|P(\mathscr{T})|$ and so $P(\mathscr{T})=P$.
\qed
\end{proof}

\subsection{An initial search DAG}

Our algorithm performs an exhaustive search in the lattice
$(X,\sqsubseteq)$ searching for constructible partitions. In the
search, we maintain and incrementally update MGTAs for every partition
we visit. First, we describe simple branching rules to obtain a rooted directed acyclic graph search structure and later give rules to prune this DAG to a node-disjoint search tree.

The root of the DAG is taken to be the initial partition
$I$ that is always constructible. For each partition $P\in X$ we next define the set $C(P)\subseteq
X$ of children of $P$. Our algorithm always proceeds by combining
parts of the partition currently being visited, so for each $P'\in
C(P)$ we will have $P'\sqsubseteq P$.
Say we visit a partition $P\in X$. We have two possibilities:
\begin{enumerate}[C1. ]
	\item \emph{$P$ is constructible:}
	\begin{enumerate}[1. ]
		\item If $P$ is not a refinement of the target pattern $P(c)$, that is if $P(c)\not\sqsubseteq P$, we can drop this branch of the search, since no possible descendant  $P'\sqsubseteq P$ can be a refinement of $P(c)$ either. (i.e. $C(P)=\tyhja$)
		\item In case $P(c)\sqsubseteq P$, we can use the MGTA for $P$ to give a concrete solution to the PATS problem instance defined by the colouring $c$. To continue the search and to find more optimal solutions we consider each pair of parts $\{p_1,p_2\}\subseteq P$ in turn and recursively visit the partition $P[p_1,p_2]$ where the two parts are combined. In fact, by the above analysis, it is sufficient to consider only pairs of the same colour:
\begin{equation}	
C(P)=\{P[p_1,p_2]\enspace|\enspace p_1,p_2\in P,\enspace p_1\neq p_2,\enspace \exists k\in P(c): p_1,p_2\subseteq k\}\enspace.
\end{equation}
	\end{enumerate}
	\item \emph{$P$ is not constructible:} In this case the MGTA $f$ for $P$ gives $f(p_1)_S=f(p_2)_S$ and $f(p_1)_W=f(p_2)_W$ for some parts $p_1\neq p_2$. We continue the search from partition $P[p_1,p_2]$:
\begin{equation}
C(P)=\{P[p_1,p_2]\}\enspace.
\end{equation}
\end{enumerate}

To guarantee that our algorithm finds the optimal solution in the case C2 above, we need the following.
\begin{lemma}
Let $P\in X$ be a non-constructible partition, $f$ the MGTA for $P$ and $p_1,p_2\in P$, $p_1\neq p_2$, parts such that $f(p_1)_S=f(p_2)_S$ and $f(p_1)_W=f(p_2)_W$. For all constructible $C\sqsubseteq P$ we have $C \sqsubseteq P[p_1,p_2]$.
\end{lemma}
\begin{proof}
Let $P$, $f$, $p_1$ and $p_2$ be as in the statement of the lemma. Let $C\sqsubseteq P$ be a constructible partition and $g:C\to\Sigma^4$ the MGTA for $C$. Since $C$ is coarser than $P$ we can obtain from $g$ a tile assignment $g':P\to\Sigma^4$ such that $g'(p)=g(q)$, where $q\in C$ is the unique part for which $p\subseteq q$. The assignment $g'$ has property \textit{A1} and so using \textit{A2} for the MGTA $f$ we get that
\begin{equation}
f(p_1)_S=f(p_2)_S \enspace\&\enspace f(p_1)_W=f(p_2)_W 
\quad\Longrightarrow \quad g'(p_1)_S=g'(p_2)_S \enspace\&\enspace g'(p_1)_W=g'(p_2)_W\enspace.
\end{equation}
Now, since $C$ is constructible, the identities $g(q_1)_S=g(q_2)_S$ and $g(q_1)_W=g(q_2)_W$ can not hold for any two different parts $q_1,q_2\in C$. Looking at the definition of $g'$, we conclude that $p_1\subseteq q$ and $p_2 \subseteq q$ for some $q \in C$. This demonstrates $C \sqsubseteq P[p_1,p_2]$. \qed
\end{proof}

\subsection{Pruning the DAG to a search tree}

Computational resources should be saved by not visiting any
partition twice. To keep the branches in our search structure
node-disjoint, we maintain a list of graphs that store restrictions on
the choices the search can make.

For each partition $P \sqsupseteq P(c)$ we associate a family of
undirected graphs $\{G^P_k\}_{k\in P(c)}$, one for each colour region
of the pattern $P(c)$. Every part in $P$ is represented by a vertex in
the graph corresponding to the colour of the part. More formally, the
vertex set $V(G^P_k)$ is taken to be those parts
$p\in P$ for which $p\subseteq k$. (So now, $\bigcup_{k\in P(c)}
V(G^P_k) = P$.) An edge $\{p_1,p_2\}\in E(G^P_k)$ indicates that the
parts $p_1$ and $p_2$ are not allowed ever to be combined in the
search branch in question. When we start our search with the initial
partition $I$, the edge sets are initially empty,
$E(G^I_k)=\tyhja$. At each partition $P$, the graphs $\{G^P_k\}_{k\in
  P(c)}$ have been determined inductively and the graphs for those
children $P'\in C(P)$ that we visit are defined as follows.

\begin{enumerate}[D1. ]
	\item \emph{If $P$ is constructible:} We choose some ordering $\{p_i,q_i\}$, $i=1,\ldots,N$, for similarly coloured pairs of parts. Define $l_i\in P(c)$, $1\leq i \leq N$ to be the colour of the pair $\{p_i,q_i\}$, so that $p_i,q_i\subseteq l_i$. Now, we visit a partition $P[p_i,q_i]$ if and only if $\{p_i,q_i\}\notin E(G^P_{l_i})$. When we decide to visit a child partition $P'=P[p_j,q_j]$, we define the edge sets $\{E(G^{P'}_k)\}_{k\in P(c)}$ as follows:
	\begin{enumerate}[1. ]
		\item We start with the graphs $\{G^P_k\}_{k\in P(c)}$ and add the edges $\{p_i,q_i\}$ for all $1\leq i < j$ to their corresponding graphs. Call the resulting graphs $\{G^\star_k\}_{k\in P(c)}$.
		\item Finally, as we combine the parts $p_j$ and $q_j$ to obtain the partition $P[p_j,q_j]$, we merge the vertices $p_j$ and $q_j$ in the graph $G^\star_{l_j}$ (After merging, the neighbourhood of the new vertex $p_j \cup q_j$ is the union of the neighbourhoods for $p_j$ and $q_j$ in $G^\star_{l_j}$). The graphs $\{G^{P'}_k\}_{k\in P(c)}$ follow as a result.
	\end{enumerate}
	\item \emph{If $P$ is not constructible:} Here, the MGTA for $P$ suggests a single child partition $P'=P[p_1,p_2]$ for some $p_1,p_2\subseteq l\in P(c)$. If $\{p_1,p_2\}\in E(G^P_l)$, we terminate this branch of the search. Otherwise, we define the graphs $\{G^{P'}_k\}_{k\in P(c)}$ to be the graphs $\{G^P_k\}_{k\in P(c)}$, except that in $G^{P'}_l$ the vertices $p_1$ and $p_2$ have to be merged.
\end{enumerate}

One can see that the outcome of this pruning process is a search tree
that has node-disjoint branches and one in which every possible
constructible partition is still guaranteed to be found.  Figure
\ref{fig:lattice} presents a sketch of the search tree.

Note that we
are not usually interested in finding every constructible partition
$P\in X$, but only in finding a minimal one (in terms of $|P|$). Next,
we give an efficient method to lower-bound the partition sizes of a
given search branch.

\begin{figure}[t]
  \centering
  \includegraphics[width=0.85\textwidth]{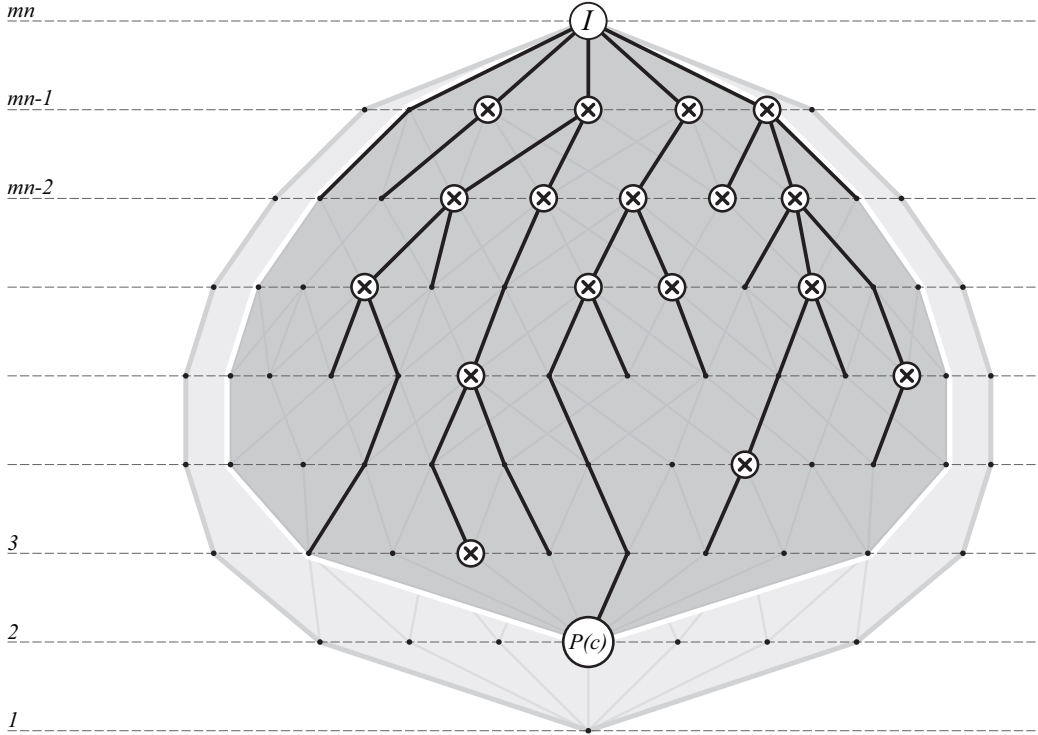}
  \caption{The search tree in the lattice $(X,\sqsubseteq)$. We start with the initial partition $I$ of size $|I|=mn$. The partition $P(c)$ defines the PATS problem instance: We search for constructible partitions (drawn as crosses) in the sublattice (shaded with darker grey) consisting of those partitions that are refinements of $P(c)$. The search tree branches only at the constructible partitions and the tree branches are node-disjoint.}
  \label{fig:lattice}
\end{figure}

\subsection{The bounding function}

Given a root $P\in X$ of some subtree of the search tree, we ask: What
is the smallest partition that can be found from this subtree?
The nodes in the subtree rooted at $P$ consists of those partitions
$P'\sqsubseteq P$ that can be obtained from $P$ by merging pairs of
parts that are not forbidden by the graphs $\{G^P_k\}_{k\in
  P(c)}$. This merging process halts precisely when all the graphs
$\{G^{P'}_k\}_{k\in P(c)}$ have beed reduced into cliques. As is well
known, the size of the smallest clique that a graph $G$ can be turned
into by merging non-adjacent vertices is given by the \emph{chromatic
  number} $\chi(G)$ of the graph $G$. This immediately gives the
following.
\begin{proposition}
For every $P'\sqsubseteq P$ in the subtree rooted at $P$ and constrained by $\{G^P_k\}_{k\in P(c)}$, we have
\begin{equation}
\sum_{k\in P(c)}\chi(G^P_k) \enspace \leq \enspace |P'|\enspace.
\end{equation}
\end{proposition}

Determining the chromatic number of an arbitrary graph is an
\textsf{NP}-hard problem. Fortunately, we can restrict our
graphs to be of a special form: graphs that consist only of a clique
and some isolated vertices. For these graphs, the chromatic numbers
are given by the sizes of the cliques.

To see how to maintain graphs in this form, consider as a base case the initial partition $I$. Here, $E(G^I_k)=\tyhja$ for all $k\in P(c)$, so $G^I_k$ is of our special form---it has a clique of size 1. For a general partition $P$, we go through the branching rules D1-D2.

\begin{enumerate}[D1: ]
	\item \emph{$P$ is constructible:} Since we are allowed to choose an arbitrary ordering  $\{p_i,q_i\}$, $i=1,\ldots,N$, for the children $P[p_i,q_i]$, we design an ordering that preserves the special form of the graphs. For a graph $G$ of our special form, let $K(G)\subseteq V(G)$ consist of those vertices that are part of the clique in $G$. In the algorithm, we first set $H_k = G^P_k$ for all $k\in P(c)$ and repeat the following process until every graph $H_k$ is a complete clique.
	\begin{enumerate}[1. ]
		\item Pick some colour $k\in P(c)$ and an isolated vertex $v \in V(H_k)\smallsetminus K(H_k)$.
		\item Process the pairs $\{v,u\}$ for all $u\in K(H_k)$ in some order. By the end, update $H_k$ to include all the edges $\{v,u\}$ that were just processed (the size of the clique in $H_k$ increases by one).
	\end{enumerate}
	A moment's inspection reveals that when the graphs $G^P_k$ are of our special form, so are all of the derived graphs passed on to the children of $P$.
	\item \emph{$P$ is not constructible:} If the algorithm decides to continue the search from a partition $P'=P[p_1,p_2]$, for some $p_1,p_2\subseteq l\in P(c)$, we have $\{p_1,p_2\}\notin E(G^P_l)$. This means that either $p_1,p_2\in V(G^P_l)\smallsetminus K(G^P_l)$, in which case we are merging two isolated vertices, or one of $p_1$ or $p_2$ is part of the clique $K(G^P_l)$, in which case we merge an isolated vertex to the clique. In both cases, we maintain the special form in the graphs $\{G^{P'}_k\}_{k\in P(c)}$.
\end{enumerate}

\subsection{Traversing the search tree}

When running a B\&B algorithm we maintain a ``current best solution''
discovered so far as a global variable. This solution gives an upper
bound for the minimal value of the tile set size and can be used to
prune such search branches that are guaranteed (by the bounding
function) to only yield solutions worse than the current best.
There are two general strategies to traverse a B\&B search tree:
\emph{Depth-First Search} and \emph{Best-First Search}~\citep{Clausen1999}.
Our description of the search tree for
the lattice $X$ is general enough to allow either of these strategies
to be used in an actual implementation of the algorithm.

In the next section we give performance data on our DFS implementation of
the B\&B algorithm.

\section{Results}
\label{sec:results}

\begin{figure}[t]
  \centering
  \includegraphics[width=\textwidth]{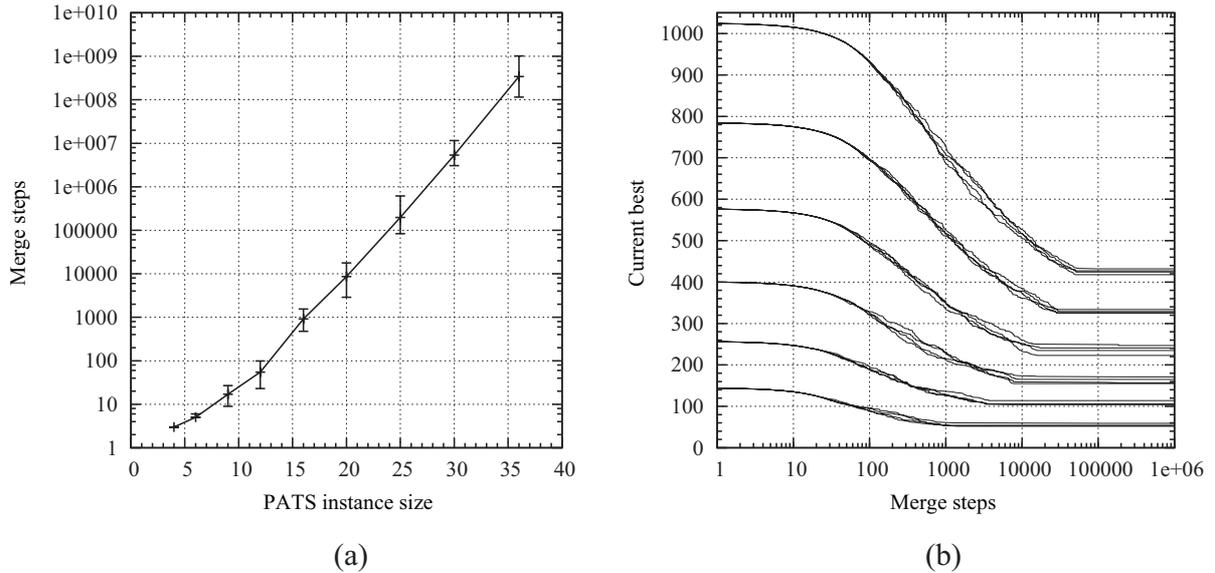}
  \caption{(a) Running time of the algorithm (as measured by the number of merge operations) to solve random 2-coloured near-square-shaped instances of the PATS problem. (b) Evolution of the tile set size of the ``current best solution'' for several large random 2-coloured instances of the PATS problem.}
  \label{fig:runtime}
\end{figure}
\begin{figure}[t]
  \centering
  \includegraphics[width=\textwidth]{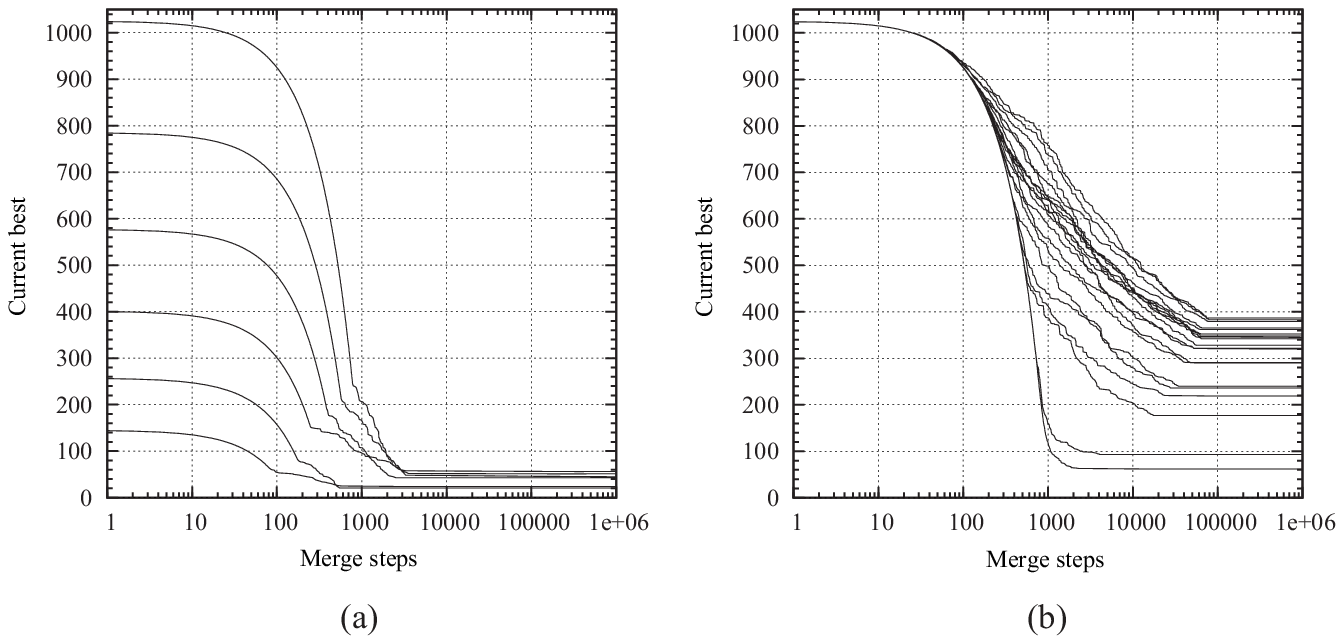}
  \caption{Evolution of the ``current best solution'' for (a) the Sierpinski pattern and for (b) the binary counter pattern. Randomization in the DFS has a clear effect on the performance of the algorithm in the case of the binary counter pattern.}
  \label{fig:curbest}
\end{figure}

The running time of our B\&B algorithm is proportional---up to a
polynomial factor---to the number of partitions the algorithm
visits. Hence, we measure the running time in terms of
the number of merge operations performed in the search.
Figure \ref{fig:runtime}a presents the running time of the algorithm
to find a minimal solution for random 2-coloured instances of the PATS
problem. The algorithm was executed for instance sizes $2\times2,
2\times 3, 3\times 3, \cdots, 5\times6$ and $6\times6$; the 20th and
80th percentiles are shown alongside the median of 21 separate runs
for each instance size. For the limiting case $6\times6$, the
algorithm spent on the order of two hours of (median) computing time
on a 2,61 GHz AMD processor.

Even though B\&B search is an exact method, it can be used to
find approximate solutions by running it for a suitable length of
time. Figure \ref{fig:runtime}b illustrates how the best solution
found up to a point develops as increasingly many steps of the
algorithm are run. The figure provides data on random 2-coloured
instances of sizes from $12\times12$ up to $32\times32$. Because we
begin our search from the initial partition, the best solution at the
first step is precisely equal to the instance size. For each size,
several different patterns were used. The algorithm was cut off after
$10^6$ steps. By this time, an approximate reduction of 58\% in the
size of the tile set was achieved (cf. a reduction of 43.5\% in
\citep{Ma2008}).

Next, we consider two well known examples of structured patterns: the
discrete Sierpinski triangle (part of which was shown in Figure
\ref{fig:sierpinski}) and the binary counter (see Figure 1 in
\citep{Rothemund2000}). A tile set of size 4 is optimal for both of
these patterns. First, for the Sierpinski pattern, we get a tile
reduction of well over 90\% (cf. 45\% in \citep{Ma2008}) in Figure
\ref{fig:curbest}a. We used the same cutoff threshold and instance
sizes as in Figure \ref{fig:runtime}b. Our description of the B\&B
algorithm leaves some room for randomization in deciding which search
branch a DFS is to explore next. This randomization does not seem to
affect the search dramatically when considering the Sierpinski
pattern---the separate single runs in Figure \ref{fig:curbest}a are
representative of an average randomized run. By contrast, for the
binary counter pattern, randomized runs for single instance size do
make a difference. Figure \ref{fig:curbest}b depicts several seperate
runs for instance size $32\times32$. Here, each run brings about a
reduction in solution size that oscillates between a reduction
achieved on a random 2-coloured instance (\ref{fig:runtime}b) and a
reduction achieved on the Sierpinski instance
(\ref{fig:curbest}a). This suggests that, as is characteristic of DFS
traversal, restarting the algorithm with a different random seed may
help with large instances that have small optimal solutions.

\section{Conclusion}
\label{sec:conclusion}

We have presented an exact branch-and-bound algorithm for finding
minimum-size tile sets that self-assemble a given $k$-coloured pattern
in a uniform self-assembly setting. Simulation results
indicate that our algorithm is able to find provably minimal tile sets for random instances of
sizes up to $6\times 6$ and can give approximate solutions for larger
instances as well.

One research direction to pursue would be to study tile sets that self-assemble an infinite,
but finite-period pattern. Does this generalization reduce easily to
the finite case? Do there exist minimal tile sets that tile the plane
aperiodically while still producing a periodic colour pattern?

\renewcommand\bibsection{\section*{References}}
\newpage
\small
\bibliographystyle{abbrv}
\bibliography{lahteet}

\end{document}